\documentclass[11pt]{article}

\usepackage{fullpage}

\usepackage{amssymb,amsthm,mathtools,enumerate,thmtools,thm-restate,bm}
\usepackage[pagebackref,colorlinks=true,urlcolor=blue,linkcolor=blue,citecolor=blue]{hyperref} 

\usepackage[nameinlink, capitalize]{cleveref}
\usepackage[vlined,boxed,commentsnumbered, ruled,linesnumbered]{algorithm2e}
\SetKwInOut{Given}{Input}
\SetKwFor{RepTimes}{Repeat}{times}{}

\usepackage{dsfont}

\usepackage[noend]{algpseudocode}

\usepackage[margin=1in]{geometry}

\usepackage[table]{xcolor}
\usepackage{booktabs}
 \usepackage{array}
 \definecolor{lightgreen}{RGB}{180,240,180}
 \definecolor{lightred}{RGB}{255,220,230}

\declaretheorem[name=Theorem, parent=section]{theorem}
\declaretheorem[name=Corollary, sibling=theorem]{corollary}

\declaretheorem[name=Lemma, sibling=theorem]{lemma}
\declaretheorem[name=Claim, sibling=theorem]{claim}

\newtheorem*{Remark}{Remark}

\theoremstyle{definition}
\declaretheorem[name=Definition, sibling=theorem]{definition}

\newcommand\R{\mathbb{R}}

\newcommand\N{\mathbb{N}}

\newcommand\F{\mathbb{F}}

\newcommand\Q{\mathbb{Q}}

\newcommand\card[1]{\left| {#1} \right|}

\newcommand\set[1]{{\left\{ #1 \right\}}}

\newcommand\Prob[2]{{\Pr_{#1}\left[ {#2} \right]}}
\newcommand\cProb[3]{{\Pr_{#1}\left[ \left. #2 \;\right\vert #3 \right]}}
\newcommand\Expc[2]{{\mathop{\mathbb{E}}_{#1}\left[ {#2} \right]}}

\newcommand\ceil[1]{\left\lceil{#1}\right\rceil}

\newcommand\floor[1]{\left\lfloor{#1}\right\rfloor}

\newcommand\cube[1]{ { \set{0,1}^{#1} } }

\newcommand\eps{\varepsilon}
\newcommand\poly{\text{poly}}

\newcommand*\xor{\mathbin{\oplus}}

\newcommand*\bxor{\mathbin{\bigoplus}}

 \newcommand{\ball}{B}
\SetKwProg{Procedure}{Procedure}{}{}
\SetKwInOut{Given}{Given}

\DeclareMathOperator{\maj}{\mathsf{maj}}

\newcommand\ind[1]{{\mathds 1}_{#1}}

\newcommand{\cD}{{\cal D}}

\newcommand{\cL}{{\cal L}}

\newcommand{\cN}{{\cal N}}

\newcommand{\cP}{{\cal P}}

\newcommand{\cT}{{\cal T}}

\newcommand{\srad}{r}

\newcommand{\oracle}[1]{\mathcal{O}_{#1}}
\newcommand{\blr}[1]{\ensuremath{\text{\sc XorTest}_{#1}}}

\newcommand{\Dist}{{\text{Dist}}}

\newcommand{\fyes}{f_{_{\text{YES}}}}
\newcommand{\fno}{f_{_{\text{NO}}}}
\usepackage{soul}
\newcommand{\short}[2]{#2}

\title{On Optimal Testing of Linearity}

\author{
{Vipul Arora\thanks{National University of Singapore. Supported in part by NRF-AI Fellowship R-252-100-B13-281. Email: \href{mailto:vipul@comp.nus.edu.sg}{\texttt{vipul@comp.nus.edu.sg}}}}
\and {Esty Kelman\thanks{Boston University and Massachusetts Institute of Technology. Supported in part by the National Science Foundation under Grant No. 2022446 and in part by NSF TRIPODS program (award DMS-2022448). Email: \href{mailto:ekelman@mit.edu}{\texttt{ekelman@mit.edu}}}} 
\and {Uri Meir\thanks{Tel-Aviv University. Email: \href{mailto:urimeir.cs@gmail.com}{\texttt{urimeir.cs@gmail.com}}}}
}

\date{\today}

\begin{document}

\maketitle

\begin{abstract} 
Linearity testing has been a focal problem in property testing of functions. We combine different known techniques and observations about linearity testing in order to resolve two recent versions of this task.

First, we focus on the online manipulations model introduced by Kalemaj, Raskhodnikova and
Varma (ITCS 2022 \& Theory of Computing 2023). In this model, up to $t$ data entries are adversarially manipulated after each
query is answered. Ben-Eliezer, Kelman, Meir, and Raskhodnikova (ITCS 2024) showed an asymptotically optimal linearity tester that is resilient to $t$ manipulations per query, but their approach fails if $t$ is too large. 
We extend this result, showing an optimal tester for almost any possible value of $t$. First, we simplify their result when $t$ is small, and for larger values of $t$ we instead use sample-based testers, as defined by Goldreich and Ron (ACM Transactions on Computation Theory 2016).
A key observation is that sample-based testing is resilient to online manipulations, but still achieves optimal query complexity for linearity when $t$ is large.
We complement our result by showing that when $t$ is \emph{very} large, any reasonable property, and in particular linearity, cannot be tested at all.  

Second, we consider linearity over the reals with proximity parameter $\varepsilon$. Fleming and Yoshida (ITCS 2020) gave a tester using
$O(1/\varepsilon\ \cdot log(1/\varepsilon))$ queries. We simplify their algorithms and modify the analysis accordingly, showing an optimal tester that only uses $O(1/\varepsilon)$ queries.
This modification works for the low-degree testers presented in Arora, Bhattacharyya,
Fleming, Kelman, and Yoshida (SODA 2023) as well, resulting in optimal testers for degree-$d$ polynomials, for any constant degree $d$.

\end{abstract}
\thispagestyle{empty} 
\newpage
\setcounter{page}{1}

\section{Introduction}

In the field of \emph{property testing}\cite{RS96, GGR98}, a randomized algorithm is given oracle access to a large object, and a promise that the object either has some property (YES instances), or is ``far'' from having it (NO instances), under some notion of distance.
The goal of the algorithm is to distinguish between the two cases with success probability at least $2/3$, where the complexity is measured by the number of oracle calls made\footnote{This is well justified, as the running time is typically polynomial in the query complexity.}.
\paragraph{Linearity testing.}
In this manuscript, we focus on testing functions over a field $\F$ (in particular $\F = \F_2$ and $\F = \R$).
That is, the tester is allowed oracle access to some function $f: \F^n \to \F$ and distinguishes functions with a desired property from ones that are far from \emph{any} function satisfying the property.
The most studied property of functions is linearity, for which a tester was first given in~\cite{BLR93}, perhaps the earliest instance of a property tester in the literature.
The test, dubbed ``the $3$ points test'' simply draws two random points $x,y \in \F^n$ and accepts only if $f(x) + f(y) = f(x \oplus y)$, where $+$ denotes addition over $\F$, and $\oplus$ denotes point-wise addition. 
By definition, if $f$ is linear, this equality must hold for any pair $x,y$.
It was shown in various settings that if $f$ is $\eps$-far from linear, there is an \emph{inequality} for at least an
$\Omega(\eps)$-fraction
of the pairs $x,y$.\cite{BLR93,BellareCHKS96,FlemingYoshida20}.
Thus, repeating the process $\theta(1/\eps)$ times defines a tester for linearity with the same query complexity. We note this tester satisfies the additional requirement of \emph{never} rejecting a linear function, making it a \emph{one-sided error} tester for linearity.

It is well known that $\Omega(1/\eps)$ oracle calls are necessary to test linearity (and most reasonable properties~\cite{fischer2024basic}), proving the 3 points test of \cite{BLR93} is optimal in terms of query complexity.

\subsection{Online manipulations over the Boolean field}

Linearity was studied extensively over the Boolean field 
$\F_2$~\cite{BLR93, BellareS94, FeigeGLSS96, BellareCHKS96, BellareGS98, Trevisan98, SudanT98,  SamorodnitskyT00, HastadW03,Ben-SassonSVW03, Samorodnitsky07, SamorodnitskyT09, ShpilkaW06, KaufmanLX10, KalemajRV22}, and an entire survey was dedicated to this problem~\cite{RasR16}.  Before we present the online manipulation model, we first review prominent ideas from previous works that we later use.

\paragraph{The $k$-point test.} 
Recently, in the context of online manipulations over $\F = \F_2$ (which we describe soon), the $k$-points test was proposed by~\cite{KalemajRV22}. In this test, which generalizes the $3$-point test of \cite{BLR93}, the algorithm chooses points $x_1,\dots, x_{k-1}$ at random, and accepts only if\footnote{We use $[k]$ for the set $\set{1,\dots,k}$.} $\sum_{i\in[k-1]} f(x_i) = f\left(\xor_{i\in[k-1]} x_i\right)$.
It was later shown by~\cite{ben2024property} that repeating the $k$-points test $\Theta(1/(k\eps))$ times produces an optimal linearity tester for any $k = O(1/\eps)$.

\paragraph{Sample-based testers.}
Linearity over $\F_2$ was also considered in~\cite{goldreich2016sample} using a weaker oracle access, where the tester cannot retrieve $f(x)$ for a value $x$ of its choice, but instead each oracle call outputs a pair $(x, f(x))$ with a uniformly random choice of $x \in_R \F_2^n$. A tester using such oracle access is called a \emph{sample-based} tester~\cite{GGR98}.
It was shown in~\cite{goldreich2016sample} that $\Theta\left(n + 1/\eps\right)$ samples are sufficient and necessary to test linearity.
The gist is that a linear function is determined by its value on a basis $B$ of the entire space $\F_2^n$. After $\Theta(n)$ samples, one can extrapolate \emph{some} linear function $g$, and use the remaining $1/\eps$ samples to compare $f$ and $g$. A formal argument for this setting, without online manipulations, appears in \autoref{sec:sample_based}.

\vspace{10mm}
The \emph{online manipulation model} was recently introduced by~\cite{KalemajRV22}.
In this model, after each query is answered, $t$ data points go through manipulations. These can be either (1) erasures, where an entry $f(x)$ is replaced with $\bot$; or (2) corruptions, where an entry $f(x)$ can be changed to any value in the range of $f$.

In particular, they study linearity testing over $\F_2$.
First, they note the $3$-points-test is fragile in the presence of an adversary: a tester that insists on querying a triplet of the form $(x, y, x+y)$ must make $\Omega\left( t \right)$ queries, since fewer queries allow the adversary to erase all (new) pairs that were created at each step.
To overcome this, they devise the $k$-points test, which creates more combinations, and eventually leads to a tester with query complexity $O(\log t / \eps)$.
It was then pointed out in~\cite{ben2024property} that the $k$-point test not only preserves the probability of spotting a violation, but actually increases it by a factor of $\approx k$ --- this is enough to reach a query complexity of $O\left(\log t + 1/\eps\right)$, which is optimal. Indeed, a lower bound of $\Omega(\log t)$ was shown for this model by\cite{KalemajRV22}, and an $\Omega(1/\eps)$ holds even without manipulations.
Another recent work by Minzer and Zheng~\cite{MinzerZ24} showed a striking $O(\log ^{\Theta(d)}(t))$ tester for degree-$d$ polynomials over any finite field. For the special case of linearity ($d=1$) over $\F_2$, however, their result is suboptimal
, achieving a query complexity of $O\left(\log^6(t/\eps)/\eps\right)$.

\paragraph{How many manipulations can the tester handle?}
All the above testers share a similar limitation: they all require $t \leq 2^{n/c}$ for some $c \geq 2$ (See~\autoref{table:online_ersures_results} for the precise values). 
It is natural to ask what is the highest value of $t$ for which linearity is still testable. Clearly, $t = 2^n$ is too much, but how close can we get?

In \autoref{sec:online_manipulations} we fully resolve this question.
Our observation, roughly speaking, is that all previous testers include some queries that are ``too predictable'', making them susceptible to adversarial manipulations.
However, this only problematic when the adversary has very large manipulation budget $t = 2^{\Omega(n)}$, in which case we can simply apply the sample-based tester instead, killing two birds with one stone.
First, sample-based testers almost automatically overcome online manipulations. Second, the seemingly high sample complexity of the tester is actually optimal as in this regime we have $\Theta(\log(t) + 1/\eps) = \Theta(n + 1/\eps)$. Formally, we have the following.

\begin{theorem}[Doubly-optimal tester]
     \label{thm:online_erasures}
   There exists a constant $c>0$ such that for all $n\in\N$, $\eps \in (0, 1/2]$ and $t \leq c\cdot \min\set{\eps^2, 1/n^2} \cdot 2^n$, 
   there exists an $\eps$-tester for linearity of multivariate functions $f:\{0,1\}^n\to\{0,1\}$ that is resilient to $t$-online manipulations budget-managing adversaries and makes  $O\left(\max\set{1/\eps, \log t}\right)$ queries. Furthermore, in the case of erasure adversary, the tester has one-sided error.        
\end{theorem}

While our tester achieves optimal query complexity, it also handles essentially the largest manipulation budget $t$ for which linearity is still testable, making it \emph{doubly optimal}.
Indeed, in~\cref{sec:impossibility_result}, we show an impossibility result for the online model with the following implication for linearity:

\begin{theorem}
    \label{thm:online_erasures_lb}
    $\eps$-testing linearity in the online manipulation  model is impossible for any $\eps \in (0, 1/2)$, provided that $t \geq 20\eps^2 2^n$.
\end{theorem}

The argument we show in~\cref{sec:impossibility_result} is actually much more general. We consider properties of functions $\F_q^n \to \F_q$, where the input size is $N := q^n$,
and show that many properties become untestable in the online manipulation model when $t$ is too large.
Trivially, $t = N$ makes testing impossible as the entire input can be manipulated after just a single query is answered authentically.
We show that many properties are untestable even when $t = \Theta(\eps^2 N)$, which is the strongest generic statement one can hope for (as it is tight with~\cref{thm:online_erasures}).

The intuition is this: consider a property that requires $\Omega(1/\eps)$, and recall $t$ is an $\eps^2$-fraction of the input size.
On the one hand, after $c/\eps$ queries (for small enough $c$) the tester still cannot distinguish certain YES and NO instances with high enough probability.
On the other hand, at this point an online adversary can already manipulate an $\eps$-fraction of the original input, enough to completely blur the initial difference between a YES and a NO instance. Hence, any additional query is useless, and the tester is doomed to fail.

To formally show this for a wide variety of properties, we leverage a recent elegant argument of~\cite{fischer2024basic}, that constructs a concrete hard distribution for any property $\mathcal{P}$ that satisfies certain ``sanity'' conditions (to be specified later).
Although our result is stated for functions $f:\F_q^n \to \F_q$, it can be easily adapted to different input types as well.

\begin{table*}
	\centering
	\begin{tabular}{|l|l|l|l|}
		\toprule
		{Result type}
		& {Query Complexity}
		& {Manipulation budget} 
		& {Reference}\\ 
		\toprule
        \midrule
        %%%%%%%%%%%%%%%%%%%%%%%%%%%%%%%%%		
	    Lower bound
		&\cellcolor{lightred}$\Omega(\log(t) + 1/\eps)$
		& any $t$
		& \cite{KalemajRV22}
		\\
	\midrule
        %%%%%%%%%%%%%%%%%%%%%%%%%%%%%%%%%		
	    Impossibility (new)
		& Untestable
		&\cellcolor{lightred}$t \geq \tilde{\Omega}\left(2^{n}\right)$
		& \Cref{thm:online_erasures_lb}
		\\
		\midrule
        \midrule
        %%%%%%%%%%%%%%%%%%%%%%%%%%%%%%%%%		
	    Algorithm
		&$O(\log(t)/\eps)$
		&$t \leq \tilde{O}\left(2^{n/4}\right)$
		& \cite{KalemajRV22}
		\\
		\midrule
		%%%%%%%%%%%%%%%%%%%%	
		Algorithm$^{\$}$
		&$O(\log^6(t/\eps)/\eps)$
		&$t \leq \tilde{O}\left(2^{n/20}\right)$
		& \cite{MinzerZ24}
		\\
        \midrule
		%%%%%%%%%%%%%%%%%%%%	
		Algorithm
		&\cellcolor{lightgreen}$O(\log(t) + 1/\eps)$
		&$t \leq \tilde{O}\left(2^{n/2}\right)$
		& \cite{ben2024property}
		\\
        \midrule
        %%%%%%%%%%%%%%%%%%%%	
		Algorithm (new)
        % New tester
		&\cellcolor{lightgreen}$O(\log(t) + 1/\eps)$
		&\cellcolor{lightgreen}$t \leq \tilde{O}\left(2^{n}\right)$
		& \Cref{thm:online_erasures}
		\\
        %%%%%%%%%%%%%%%%%%%%
  
		\bottomrule
	\end{tabular}
	~
	\caption{ 
	      \small A comparison of all known results for linearity in the online manipulation model.
         For the sake of clarity, bounds on $t$ hide factors of $\poly(n/\eps)$.
        Colored cells mark optimal lower and upper bounds, emphasizing the new algorithm is doubly-optimal.
        $^{\$}$The algorithm by~\cite{MinzerZ24} is devised for a more general setting of low-degree testing over $\F_q$. Its application to our case ($d = 1, q = 2$) is suboptimal and included only to give the full picture.
		}
\label{table:online_ersures_results}
\end{table*}

\subsection{Testing Over the Reals}\label{sec:intro-reals}

We start by describing the model of testing functions over the reals.
It is similar to the case of finite domains, with the important distinction that a uniform measure over the domain no longer exists. 
This is indeed an issue, as the uniform measure over finite domains is crucial to most standard notions of distance (e.g., hamming distance).
Instead, it is natural to choose the standard (multivariate) Gaussian, denoted by $\cN(\bm 0,I)$.
More generally, the \emph{distribution-free} testing framework \cite{GGR98,halevyK07}, considers a distance measure with respect to an unknown but samplable distribution over the domain.
Formally, for a measurable function $f:\R^n\to\R$, a proximity parameter $\eps>0$, an unknown but samplable distribution $\cD$ supported over $\R^n$, and a property $\cP$, $f$ is $\eps$-far from $\cP$ with respect to $\cD$ if\footnote{We use bold face, e.g. $\bm x$, to denote vectors/points in $\R^n$.}
\[\delta_{\cD}(f,\cP) :=  \inf_{g\in\cP}\left\{ \Pr_{\bm x\sim\cD}[ f(\bm x) \neq g(\bm x)] \right\}>\eps.\]

An algorithm is a \emph{distribution-free tester} for $\cP$, if given query access to $f$, sample access to $\cD$, and $\eps>0$, it accepts all $f\in\cP$ with probability at least $2/3$, and rejects all $f$ such that $\delta_{\cD}(f,\cP)>\eps$, with probability at least $2/3$.

\paragraph{Linearity vs. additivity.}
Over the reals, unlike the finite case, we have a distinction between
\begin{itemize}
\item \emph{additivity}: for all $\bm x,\bm y\in\R^n,f(\bm x)+f(\bm y)=f(\bm x+\bm y)$, and
    \item \emph{linearity}: $f(\bm x)\equiv \sum_{i=1}^n c_ix_i$, for some $\{c_i\in\R,i\in\{1,\dots, n\}\}$.
\end{itemize} 
While the two are equivalent over finite fields, some pathological examples show that functions over the reals can be additive but not linear \cite{Ham05}.
In fact, over finite domains this equivalence is key for testing results, as the $3$-point BLR test directly checks for additivity.
While the tester and its analysis are similar for both properties, linearity requires additional assumptions (see~\cite{FlemingYoshida20}). The discussion below is therefore restricted to additivity.

In order to describe previous and new testing results over the reals, we first review another useful concept from the testing literature. 

\paragraph{Self-correction.}
The notion of self-correction dates back to~\cite{BLR93} and 
has deep ties with many testing results.
To demonstrate this viewpoint, we revisit the $3$-point linearity tester over $\F_2$.
Define the local correction of $p$ according to $x$ as $g_x(p) := f(p+x) - f(x)$, That is, by querying $x$ and $p+x$ we get an indication what the value of $p$ ``should be'' in order to satisfy linearity.
Define the global correction to be
\[
    g(p) := \maj_{x} g_x(p) .
\]
Think of $g$ as a function that ``corrects'' each $p$ by assigning it the value that satisfies linearity for more choices of $x$. 
Formally the $3$-point tester asks whether $f(p) = g_x(p)$, but morally it asks whether $f \equiv g$, by comparing them on a random input $p$ and replacing $g$ with a ``representative'' $g_x$ which can be computed by $2$ queries to $f$.
Denoting the rejection probability of the test by $\mu$, the tester is analyzed by showing two claims: (1) if $\mu$ is small then $g$ and $f$ are close; and (2) if $\mu$ is small then $g\in \cP$.
Combining the claims, a small rejection probability implies that $f$ is close to linear (specifically to $g$), which guarantees the soundness of the tester (by counter positive).
To show (1) and (2), a crucial condition is that evaluation of $g_x(p)$ uses queries to $f$ that have the same marginal distribution (here $x$ and $p+x$ are both marginally uniform over the domain).

\paragraph{Self-correcting over the reals.}
It turns out that self-correcting a function over the reals is tricky in particular due to the marginals condition.
To this end, ~\cite{FlemingYoshida20} defined their $2$ queries linearity self-corrector in two steps. First, in a straightforward manner inside a small ball $B$ centered at the origin, where one can show the marginals condition (approximately) holds, and the analysis follows through.
The second step is extrapolating $g$ to the entire domain using the fact that an additive function respects multiplication by a rational scalar, i.e., for all $a\in \Q$ it holds that $g(ax) = a\cdot g(x)$.  

We leave the hairy details and precise definitions for~\cref{sec:real testing}, and merely point out that due to this nuisance, evaluation of $g$ is explicit in their algorithm (unlike the $3$-point test for $\F_2$ where it is implicit).
Eventually, their test asks whether $f \equiv g$ using $\Theta(1/\eps)$ comparison points, but each evaluation of $g$ is done using $\Theta\left(\log (1/\eps)\right)$ queries to $f$, resulting in query complexity $\Theta\left((1/\eps) \log (1/\eps) \right)$.

\paragraph{Our modifications and results.}
The soundness analysis is made via the self-correction paradigm: if $f$ is $\eps$-far from $\cP$, it is also $\eps$-far from $g$ (recall the distance is according to the samplable distribution $\cD$). Each comparison point sampled from $\cD$ thus has probability $\Omega(\eps)$ to indicate that $f \not\equiv g$. Taking $O(1/\eps)$ points amplify the success probability to, say, $9/10$.
However, the reason ~\cite{FlemingYoshida20} use $\Theta\left(\log (1/\eps)\right)$ queries to $f$ for each value $g(\bm p)$ is merely to reduce the error probability from a constant to $O(\eps)$.
This way, one can apply a union bound over all $O(1/\eps)$ comparison point, adding only an additional small constant error to the tester.
Our observation is that the algorithm can tolerate erroneous evaluations, thus reducing the error for all evaluations of $g$ is is unnecessary.
Indeed, when analyzing the soundness, we simply say that $\bm p$ is a witness for $f\not\equiv g$ if it satisfies both $f(\bm p) \neq g(\bm p)$ and the evaluation of $g$ succeeded. This still gives an $\Omega(\eps)$ probability for a witness, and $O(1/\eps)$ samples from $\cD$ suffice. We do not care that other evaluations of $g$ (roughly half of them) contain errors, only the one important point matters.
We get the following result.

\begin{theorem}\label{thm:additivity-and-linearity-real-tester}
    Fix $\eps\in(0,1)$.
    There exist distribution-free, one-sided error $O(1/\eps)$-query testers for both additivity of functions $f\colon\R^n\to\R$ and linearity of continuous functions $f\colon\R^n\to\R$.
\end{theorem}

\paragraph{Low-degree testing over the reals.}
We observe that a similar modification applies to the distribution-free, low-degree tester of \cite[Theorem 1.1]{ABFKY23}, shaving off the $\log\frac{1}{\eps}$ factor from their query complexity (which is $O(d^5+\frac{d^2}{\eps}\log \frac{1}{\eps})$), resulting in an optimal $O(1/\eps)$-query complexity tester for the constant degree-$d$ regime. 
\begin{theorem}\label{thm:exact-low-degree}
    Let $d\in\mathbb{N}$, and for $L>0$, suppose $f: \R^n \to \R$ is a function that is bounded in the ball $\ball(\bm{0}, L)$. 
There exists a distribution-free, one-sided error,  $O(d^5+d^2/\varepsilon)$-query tester for testing whether $f$ is a degree-$d$ polynomial, or is $\eps$-far from all degree-$d$ polynomials over the unknown distribution $\mathcal{D}$.
\end{theorem}
 We prove \autoref{thm:exact-low-degree} in \autoref{sec:low-degree over the reals}. We note that applying the same modification to the corresponding algorithms and analysis, gives the same improvement to the approximate, and the discrete low-degree testers given by \cite{ABFKY23} for the respective settings of approximate testing (\cite[Theorem 1.2]{ABFKY23}) and testing over rational lattices (\cite[Theorem 1.3]{ABFKY23}). We state the corresponding improved results in \autoref{sec:low-degree over the reals}, and omit their proofs because they are too similar to the proof of \autoref{thm:exact-low-degree}.

\paragraph{Organization.} We next introduce some basic notation. \autoref{sec:online_manipulations} is devoted to the online model, wherein we provide preliminaries for online testing, as well as prove \Cref{thm:online_erasures,thm:online_erasures_lb}. We prove \autoref{thm:additivity-and-linearity-real-tester} in \autoref{subsec:real linearity} and \autoref{thm:exact-low-degree} in \autoref{sec:low-degree over the reals}.

\section{Testing over \texorpdfstring{$\F_2$}{F2} with online manipulations}
\label{sec:online_manipulations}
The online manipulation model was introduced by Kalmaj et al. \cite{KalemajRV22} and formally defined by Ben-Eliezer et al.\cite{ben2024property}.
For the many motivations of the model, we refer to both these papers.
We next define the model and then proceed to prove~\cref{thm:online_erasures} and~\cref{thm:online_erasures_lb}.

In the online manipulation testing model, the input is accessed via a sequence $\set{\oracle{i}}_{i\in {\mathbb N}}$ of oracles, where $\oracle{i}$  is used to answer the $i$-th query. The oracle $\oracle{1}$ gives access to the original input (e.g., when the input is a function $f$, we have $\oracle{1} \equiv f$), and subsequent oracles are objects of the same type as the input (e.g., functions with the same domain and range). Each such oracle is obtained by the adversary by modifying the previous oracle to include a growing number of erasures/corruptions as $i$ increases.
We use $\Dist(\oracle{},\oracle{}')$ for the Hamming distance between the two oracles (i.e., the number of queries for which they give different answers).
We let $t \in \R_{\geq 0}$  denote the number of \emph{erasures} (or \emph{corruptions}) \emph{per query}. 
More formally, we define:
\begin{definition}[Types of adversaries]\label{def:fixed-rate_budget-managing}
Fix $t \in \R_{\geq 0}$ to be the manipulation parameter, and let $\oracle{} = \set{\oracle{i}}_{i\in {\mathbb N}}$ be a sequence of oracles\footnote{All our results regrading the model only consider a finite prefix of this sequence.}, where $\oracle{1}$ is an oracle accessing the original input function $f$.
We have the following attributes of adversaries, each imposing a condition on the oracle sequence:
    \begin{itemize}
        \item A \emph{fixed rate} adversary must choose entries to manipulate immediately after a query is answered. Formally, for all $i\in\N$ it holds that
         \[
            \Dist(\oracle{i}, \oracle{i+1}) \leq \floor{(i+1)\cdot t} - \floor{i\cdot t}.
         \]
    
        \item A \emph{budget-managing} adversary can reserve its ``manipulation rights'' for later rounds, but it never exceeds $t$ manipulations per each query made. Formally, for all $i\in\N$ it holds that 
         \[
            \Dist(\oracle{1}, \oracle{i+1}) \leq i\cdot t.
         \]
    
        \item An \emph{erasure} adversary can only change original values of $f$ to a new special symbol $\bot$. Formally, for all $i\in\N$ it holds that
        $\oracle{i+1}(x) \in \set{\oracle{i}(x), \perp}$.
    
        \item A \emph{corruption} adversary can change values freely, but we assume it never uses values that aren't viable outputs, as these immediately give away that the entry was manipulated.
    \end{itemize}

\end{definition}

Each adversary uses either erasures or corruptions, and each adversary either has fixed rate or can manage its manipulation budget.
We note that budget-managing adversaries are at least as strong as ones with fixed rate, making it harder to test against them.

\begin{definition}[Online $\eps$-tester]\label{def:online_tester} Fix $\eps\in(0,1).$
    An online $\eps$-tester $\cT$ for a property $\mathcal{P}$ that works in the presence of a specified adversary (e.g., $t$-online erasure budget-managing adversary) is given access to an input function $f$ via a sequence of oracles 
    $\oracle{} =  \set{\oracle{i}}_{i\in {\mathbb N}}$
     induced by that type of adversary.
    For all adversarial strategies of the specified type,
    \begin{enumerate}
        \item if $f \in \mathcal{P}$, then $\cT$ accepts with probability at least 2/3, and
        
        \item if $f$ is $\eps$-far from $\cP,$
        then $\cT$ rejects with probability at least 2/3, 
    \end{enumerate}
    where the probability is taken over the random coins of the tester.

    If $\cT$ always accepts all functions $f\in\cP$, then it has \emph{one-sided error.} If $\mathcal{T}$ chooses its queries in advance, before observing any outputs from the oracle, then it is \emph{nonadaptive}.
\end{definition}

To ease notation, we use $\oracle{}(x)$ for the oracle's answer to query $x$ (omitting the timestamp $i$). 
If $x$ was queried multiple times, $\oracle{}(x)$ denotes the first answer given by the oracle.

\begin{Remark}
    All our results consider both erasures and corruptions.
    Our linearity tester is resilient even to the stronger budget-managing adversaries, while our impossibility results in~\cref{sec:impossibility_result} hold even with the weaker fixed-rate adversaries. 
\end{Remark}

\paragraph{The tester.}
As mentioned earlier, our tester uses one of two entirely different strategies, depending on the parameters $t, \eps$. We define the next important parameter, used to decide which strategy to employ:
\[
    m := 4\ceil{\log(t) + 10/\eps} .
\]
When $m \leq n/3$ (Case I), meaning that $\eps$ is not too small and $t$ is not too big, we can safely use the strategy of~\cite{ben2024property} which eventually uses the $k$-point tester.
Originally, handling a wide range of $t, \eps$ led to a very subtle choice of $k$ and longer analysis. Our focus on a certain regime allows us to choose $k = m$ and significantly simplify the analysis.
When $m > n/3$ (Case II), however, we revert to sample-based testing, showing this strategy easily defeats the adversary and coincides with the optimal query complexity.
We analyze each of the two cases separately.

\subsection{Case I: \texorpdfstring{$m \leq n/3$}{m < n/3}}
\label{sec:k-pts-test}
For this case we use the following primitive from~\cite{KalemajRV22}:
\begin{algorithm}
 \caption{$\blr{k}$
 }
 \label{alg:krv-linearity-test}
{
    \Given {Even integer parameter $k\geq 2$ and query access to a function $f:\{0,1\}^n\to\{-1, 1\}$ 
    }
     Query $k$ points $x_1,\dots, x_k \in \cube{n}$ chosen uniformly at random (with replacement).\\
     Query point $y = \bxor_{i\in[k]} x_i$.\\
     {\bf Accept} if $f(y) = \prod_{i\in[k]} f(x_i)$ (equivalently, if $f(y) \cdot \prod_{i\in[k]} f(x_i) = 1$); otherwise, {\bf reject}.
    }
\end{algorithm}  

An improved soundness guarantee for this test was shown by~\cite{ben2024property}:
\begin{lemma}[Soundness]     \label{lem:krv_k_test}
    If $f$ is $\eps$-far from  linear, and $k\geq 2$ is even, then
    \begin{align*}
        \Prob{}{\blr{k}(f)  \text{\em\ rejects}}
        \geq \frac{1 - (1-2\eps)^{k-1}}{2} 
        \geq \min\Big\{\frac 1 4 \, ,\, \frac{k \eps}{2}\Big\}.
    \end{align*}   
\end{lemma}

The algorithm is as follows:

\begin{algorithm}
 \caption{Online-Erasure-Resilient Linearity Tester}\label{alg:linearity-main}
{
    \Given {Parameters $\eps\in(0,1/2],t\in\N$; query access to $f$ via $t$-erasure oracle sequence $\oracle{}$ 
    }
     \RepTimes{$6$}  {\label{step:linearity-repeat}
         Sample $X =(x_1,\dots,x_{m})\in(\cube{n})^m$ uniformly 
         at random.\\
         Query $f$ at points $x_1,\dots, x_m$.\\
         Query $f$ at point $y = \xor_{j\in S} x_j$, where $S$ is a uniformly random subset of $[m]$ of size $\frac m2$.\\ 
         \If{$\oracle{}(y) \cdot \Pi_{j\in S} \oracle{}\left(x_j\right) = -1$}{\textbf{Reject}
         \Comment{This implies no erasures in this iteration.}
         }
        
    }
        \textbf{Accept}
  }  
 \end{algorithm}

Algorithm~\ref{alg:linearity-main} 
performs a simplified version of the tester of~\cite{ben2024property}.
We upper bound the probability of seeing an erasure at any given iteration, simplifying the analysis of~\cite{ben2024property} for this restricted regime (in particular $\eps \geq \Omega(1/n)$ and $t \leq 2^{O(n)}$).

\begin{lemma}
    \label{lem:krv_case_no_erasures}
    If $m \leq n/3$, the probability that one specific iteration of the loop in Line~\ref{step:linearity-repeat} of Algorithm~\ref{alg:linearity-main} queries an erased point is at most $3/64$. 
\end{lemma}

\begin{proof}
    If $m \leq n/3$, the algorithm enters the loop and queries a total of $6(m+1)$ points, which induce at most $6t(m+1)$ manipulations through the entire execution.
    We define three bad events and give upper bounds on their probabilities.
        \paragraph{An erasure while querying $X$.}
        Let $B_1$ be the event that $\oracle{}(x_j) = \perp$ for some point $x_j$ sampled in this iteration, where $j\in[m]$. Each point $x_j$ is sampled uniformly from $\cube{n}$, so the probability it is erased is at most $6t(m+1)/2^n$. By a union bound over all $m$ points, using $m \leq n$, we have 
        \begin{equation}
            \label{eq:bound_b1}
            \Prob{X}{B_1}\leq \frac{6t(m+1)m}{2^{n}} 
            \leq \frac{6t(m+1)m}{2^{m}}.
        \end{equation}
       
        \paragraph{$X$ induces a bad distribution of $y$ points.}
        For any choice $S$ denote $y_S = \xor_{j\in S} x_j$, and let $B_2$ be the event that, in this iteration, there exist two different choices of $S$ leading to the same choice $y \in \cube{n}$. 
        For any two distinct sets $T_1,T_2\subset [m]$ of size $m/2$ w.l.o.g.\ there exists an element $\ell\in T_1 \setminus T_2$. Fix all entries in $X$ besides $x_\ell$. The value of $y_{_{T_2}}$ is now fixed, but over the random choice of $x_\ell\in\cube{n}$, the vector $y_{_{T_1}}$ is uniform over $\cube{n}$. Thus, $\Pr_{x_{\ell}}{[y_{_{T_1}} = y_{_{T_2}}]}=2^{-n}$ and, consequently,
        \[
            \Prob{X}{y_{_{T_1}} = y_{_{T_2}}}
            = \Expc{}{\Prob{x_{\ell}}{y_{_{T_1}} = y_{_{T_2}}}}
            =\Expc{}{2^{-n}}
            = 2^{-n},
        \]
        where both expectations are over all entries in $X$ besides $x_\ell$, which are drawn independently from $x_{\ell}$.
        We use a union bound over all pairs of subsets $T_1$ and $T_2$, and the fact $m \leq n/3$ to get 
        \begin{equation}
        \label{eq:bound_b2}
            \Prob{}{B_2} = \Prob{X}{\exists T_1\neq T_2 \text{\rm\ such that\ } y_{_{T_1}} = y_{_{T_2}}}
            \leq \frac{2^{2m}}{2^n}\leq 2^{-m} .
        \end{equation}
        
        \paragraph{An erasure on query $y$.}
        Let $B_3$ be the event that $\oracle{}(y) = \perp$ for the point $y$ queried in this iteration.
        The adversary knows $X$ before $y$ is queried, but there are plenty of choices for $y$.
        Conditioned on $\overline{B_2}$, the distribution of $y$ is uniform over $\binom{m}{m/2}$ different choices. 
        We use $\binom{m}{m/2} \geq 2^m/\sqrt{2m}$, to obtain
    \begin{equation}
        \label{eq:bound_b3}
            \cProb{}{B_3}{\overline{B_2}}
            \leq \frac{6t(m+1)}{\binom{m}{m/2}}
            \leq \frac{6t(m+1)\sqrt{2m}}{2^m}.
    \end{equation}
     
    In terms of the bad events, we wish to bound $\Pr[B_1\cup B_3]$.
    By using a union bound over $B_1$ and $B_3$ and then the law of total probability to compute $\Pr[B_3]$, we get
       \begin{align*}
           \Prob{}{B_1\cup B_3}
           &\leq \Prob{}{B_1} + \Prob{}{\overline{B_2}} \cdot \cProb{}{B_3}{\overline{B_2}} + \Prob{}{B_2} \cdot \cProb{}{B_3}{B_2} \\
           &\leq \Prob{}{B_1}+\cProb{}{B_3}{\overline{B_2}} + \Prob{}{B_2} .
       \end{align*}
       We next combine the bounds from \eqref{eq:bound_b1}, \eqref{eq:bound_b2} and \eqref{eq:bound_b3}, and use $t \leq 2^{\frac{m}{4} - 10}$ (which is implied by $\eps \leq 1/2$).
       \begin{align*}
           \Prob{}{B_1\cup B_3}
           \leq \frac{16t(m+1)m}{2^m}
           \leq \frac{(m+1)m}{64 \cdot 2^{\frac{3m}{4}}}
           \leq \frac{3}{64} ,
       \end{align*}
       where the last transition holds since $\frac{(m+1)m}{ 2^{\frac{3m}{4}}} \leq 3$ for all positive values of $m$.
\end{proof}

We are ready to prove the correctness of \autoref{alg:linearity-main}, showing \autoref{thm:online_erasures} for case I.
\begin{proof}
    Algorithm~\ref{alg:linearity-main} makes $6(m+1) = O(\log t + 1/\eps)$ queries.
        
    We next analyze the algorithm in the presence of erasures. It is easy to see the algorithm always accepts all linear functions.  
    Now, fix an adversarial (budget-managing) strategy and suppose that the input function is $\eps$-far from linear. By \Cref{lem:krv_k_test} and since $k = m/2 \geq 1/\eps$ is even, the probability that one iteration of the loop in Step~\ref{step:linearity-repeat} samples a witness of nonlinearity is at least $\min\set{1/4\ ,\  k\eps/2} \geq \min\set{1/4\ ,\  1/2} = 1/4$.

    By \Cref{lem:krv_case_no_erasures}, the probability that an erasure is seen in a specific iteration is at most $1/16$.
    By a union bound, the probability of a single iteration seeing an erasure or not selecting a witness of nonlinearity is at most $1 - \frac{1}{4} + \frac{1}{16} = 1 - \frac{3}{16}$.
    Algorithm~\ref{alg:linearity-main} errs only if this occurs in all iterations. By independence of random choices in different iterations, the failure probability is at most
    \[
        \Big(1-\frac{3}{16}\Big)^6
        \leq e^{-\frac{18}{16}}
        \leq \frac{1}{3} ,
    \]
    where we used $1-x\leq e^{-x}$ for all $x$.

    Finally, we show that Algorithm~\ref{alg:linearity-main} has two-sided error at most $1/3$ in the presence of corruptions.
    The soundness holds with the same analysis, as finding a single violation suffices for this case.
    For completeness, note that the algorithm can only err if it has seen a manipulation, and the probability of seeing a manipulation at any iteration is at most $3/64$ by \Cref{lem:krv_case_no_erasures}. Using a union bound over all $6$ iterations, the overall probability of seeing any manipulated entry during the entire execution is at most $6 \cdot (3/64) \leq 1/3$. 
    \renewcommand{\qedsymbol}
   {\ensuremath{\square}(Case I)}
\end{proof}

\subsection{Case II: \texorpdfstring{$m > n/3$}{m > n/3}}
\label{sec:sample_based}

In this case, we use \cite[Theorem 5.1]{goldreich2016sample} for the domain $\cube{n}$ and run an algorithm that only uses random samples, their algorithm is simple, and we bring it here for completeness: 

\begin{algorithm}
 \caption{The Goldreich-Ron sample-based tester}\label{alg:linearity-GR sampled based}
{
    \Given {Parameters $\eps\in(0,1/2]$; sample access to $(x,f(x))$}
    $m=O(n)$.\\
    Sample $X =(x_1,\dots,x_{m})\in(\cube{n})^m$ uniformly 
         at random.\\
       \textbf{Accept} if $span(X)\neq \cube{n}$. \Comment{w.h.p. this doesn't happen;}\\ 
       Let $Y\subset X$ be an arbitrary basis for $\cube{n}$ and $g:\cube{n}\to\cube{}$ the unique linear function that agrees with $f$ on all points in $Y$.\\
      \For{$O(1/\varepsilon)$ times}{
      Sample $\bm z \sim \cube{n}$.\\
     \textbf{Reject} if $f(z) \neq g(z)$.
      }
      \textbf{Accept}.

  }  
 \end{algorithm}

\begin{lemma}[\protect{\cite[Theorem 5.1 (1)]{goldreich2016sample}}, for Boolean functions]
There exists a one-sided error sample-based tester for linearity over $\F_2$, with sample complexity $O(1/\eps+n)$.
\end{lemma}

Next, we show that if the algorithm uses only uniform random samples, then the adversary loses its power, in the sense that we, most likely, won't see any manipulation. 
\begin{theorem}
    Let $\cal{T}$ be a sample-based tester (for some property $\cal{P}$) with input length $N$ and distance parameter $\eps$ that uses $q$ uniformly random samples and succeeds with probability $1-\delta$. Then the same tester with the same number of queries succeeds with probability $1-2\delta$ in the presence of $t$-online-manipulations budget-managing adversary for any $t\leq \delta\cdot N/q^2$.  
\end{theorem}

\begin{proof}
   To analyze the tester, we note that all queries in this test are random samples, and the total number of manipulations made is $qt$. Therefore, each sample has probability of at most $qt/N$ to be a previously-manipulated point.
   By a union bound over all the samples taken, the overall probability to sample any previously-manipulated point is at most $q^2 t /N\leq \delta$.
   Applying a union bound for this event and the event that $\cT$ errs leaves a total error of at most $2\delta$., completing the proof.
\end{proof}
\begin{Remark}
If $\cal{T}$ is a one-sided error tester, we can keep it one-sided  against online-erasure adversary by accepting whenever we see an erasure. Otherwise, we get a two-sided error tester.  
\end{Remark}
\begin{corollary}
There is a one-sided error sample-based tester for testing linearity in the presence of online-manipulation budget-managing adversaries. The tester uses $q=O(1/\eps+n)$ queries, succeeds with probability $2/3$ and works for all $\eps\in (0,1/2)$ and $ t\leq O(\min\set{2^n/n^2, \eps^2 2^n})$.
\end{corollary}

We are now ready to prove \autoref{thm:online_erasures} for Case II:
\begin{proof}
    The query complexity is $q = O(n + 1/\eps)$, and in our case, it is at most $O(m + 1/\eps) = O(\log t + 1/\eps)$.
    \renewcommand{\qedsymbol}{\ensuremath{\square}(Case II)}
\end{proof}

\subsection{A testing impossibility result}
\label{sec:impossibility_result}
In this section, we prove that linearity testing is impossible when $t$ is too large (\autoref{thm:online_erasures_lb}). In fact, we show the same is true for a large variety of properties.
The argument below is formulated for properties of functions $f: \F_q^n \to \F_q$, though it can be adjusted to other settings as well (e.g., properties of length $n$ strings over alphabet $\Sigma$).  

We first review the argument of~\cite{fischer2024basic} in the standard model (with no adversary) and modify it for our input type, establishing some notations along the way. We then turn to state and prove the new impossibility result for the online manipulation model.

\paragraph{The standard model (Fischer's argument~\cite{fischer2024basic}).}

We consider properties of functions\footnote{The original argument in~\cite{fischer2024basic} is applied to properties of $\Sigma^n$, or equivalently functions $f:[n] \to \Sigma$.}
$f:\F_q^n \to \F_q$, where a property $\mathcal{P}$ is identified with a subset of all these functions.
Fix a property $\mathcal{P}$ such that there exists two input functions $\fyes \in \mathcal{P}$ and $\fno$ that is $\alpha$-far from $\mathcal{P}$ for some constant $\alpha > 0$. We focus on proximity parameter $\eps \in (0,\alpha)$, and for simplicity, assume that $\eps = \ell / q^n$ for some integer $\ell \in \N$.

First, we consider a function $g\in\mathcal{P}$ that minimizes the distance to $\fno$ (if such $g$ is not unique, choose one arbitrarily) and
let $D=\{x\in \F_q^n: g(x)\neq \fno(x)\}$ be the set of inputs on which $g$ and $\fno$ disagree.
For any $A\subseteq D$, define
\[
    g_{_A}(x) :=
    \begin{cases}
    g(x), & \text{if $x \notin A$}\\
    \fno(x), & \text{if $x \in A$}
    \end{cases}
\]
In particular, $g_{_{D}} \equiv \fno$, and $g_{_{\emptyset}} \equiv g$. Moreover, $g_{_A}$ has distance exactly $\card{A} / q^n$ from the property $\mathcal{P}$ (since $g$ minimizes the distance of $\fno$ from $\mathcal{P}$, it must minimize the distance of $g_{_A}$ as well).

The focus from now on is on a task we call $(g, D, \ell)$-testing. An algorithm solves this task if it accepts $g$ with probability at least $2/3$, but accepts $g_{_A}$ with probability at most $1/3$ for any $A \subseteq D$ such that  $\card{A} = \ell$.
In particular, for $\eps = \ell/q^n$, an $\eps$-tester for $\mathcal{P}$ is a $(g, D, \ell)$-tester.

It was shown in~\cite{fischer2024basic} that w.l.o.g, a randomized $(g, D, \ell)$-tester is non-adaptive and only queries points in $D$. Thus, any such tester simply queries the function on $m$ points $x_1, \dots x_m$, where each $x_i$ is a random variable supported on $D$.
Furthermore, for each point $z\in D$ we can define the event $E_z$ that the point $z$ was queried during the execution (i.e., $x_i = z$ for some $i\in [m]$).

Lastly, fix a randomized $(g, D, \ell)$-tester $\cT$, and observe the probabilities for any point $z\in D$ to be queried (that is, the probability of the event $E_z$).
If the tester makes at most $m$ queries, then
\[
    \sum_{z\in D} \Prob{}{E_z} 
    = \sum_{z\in D} \Expc{}{\ind{E_z}}
    = \Expc{}{\sum_{z\in D} \ind{E_z}}
    \leq \Expc{}{m} 
    = m.
\]
The crux of the proof is that $\cT$ chooses its queries ahead, and we can focus on the points least likely to be queried in order to single out an input that it struggles to reject. 
Formally, denote by $D' = \set{z_1, \dots, z_{\ell}}$ the set with $\ell$ points that have the smallest probabilities (breaking ties arbitrarily), and observe the probability of the event $E_{D'}$ that \emph{any} point from $D'$ was queried. By a union bound, we have 
\[
    \Prob{}{E_{D'}}
    \leq \sum_{z\in D'} \Prob{}{E_z}
    \leq \frac{\ell}{\card{D}} \cdot \sum_{z\in D} \Prob{}{E_z}
    \leq \frac{\ell \cdot m}{\card{D}} .
\]
We recall that $\ell = \eps \cdot q^n$ and that $\fno$ is $\alpha$-far from $\mathcal{P}$, meaning $\card{D} \geq \alpha \cdot q^n$. Hence
\begin{equation}
    \label{eq:online_lb_crucial_event_upper}
    \Prob{}{E_{D'}}
    \leq \frac{\eps \cdot m}{\alpha} .
\end{equation}
To finish up, fix the random coins of $\cT$ and consider two executions, one when fed values from $g$ and another when fed values from $g_{D'}$. 
The output can only differ if the event $E_{D'}$ occurred (otherwise, $\cT$ sees the same values in both executions).
So over the choice of random coins we have
\begin{equation}
    \label{eq:online_lb_crucial_event_lower}
    \Prob{}{E_{D'}} \geq \Prob{}{\cT\text{ accepts }g} - \Prob{}{\cT\text{ accepts }g_{D'}} \geq \frac{2}{3} - \frac{1}{3} = \frac{1}{3} .
\end{equation}
Combining~\eqref{eq:online_lb_crucial_event_upper} and~\eqref{eq:online_lb_crucial_event_lower} we $m \geq \alpha/(3\eps)$.

\paragraph{The online manipulation model.}
For this model, we similarly assume there exists $\fno$ that is $\alpha$-far, and a function $g\in \mathcal{P}$ that minimizes distance to $\fno$ (arbitrarily chosen if not unique), where $\fno$ and $g$ differ on the set of inputs $D \subseteq \F_q^n$.
We again focus on the task of $(g, D, \ell)$-testing, and assume $\eps = \ell/q^n \in (0, \alpha)$.

For the formal statement, we denote $\mathcal{F}_n := \set{f: \F_q^n \to \F_q}$ and $\mathcal{F} = \cup_{n\in \N} \mathcal{F}_n$. A property is simply $\mathcal{P} \subseteq \mathcal{F}$, and we use $\mathcal{P}_n = \mathcal{P} \cap \mathcal{F}_n$. 
The following holds
\begin{theorem}\label{thm: fisher lb for online model}
    Fix a constant $\alpha > 0$, and a property $\mathcal{P} \subseteq \mathcal{F}$, such that for infinitely many values of $n$ there exist $\fyes^n \in \mathcal{P}_n$ and $\fno^n\in \mathcal{F}_n$ that is $\alpha$-far from $\mathcal{P}_n$.
    Fix $\eps = \ell / q^n < \alpha$ where $\ell \in \N$, and $t \geq (10/\alpha) \eps^2 q^n$.
    Then there is no $\eps$-tester for $\mathcal{P}$ that is resilient to $t$ manipulations per query.
\end{theorem}

\begin{proof}
    For infinitely many values of $n$ there exists $\alpha$-far input $\fno^n$.
    Fix such $n$ and follow the argument as before, focusing on $g\in \mathcal{P}_n$ that minimizes distance to $\fno^n$, and the set of coordinates $D$ on which they differ.
    Assume there exists a $t$-manipulation-resilient $\eps$-tester $\cT$ for $\mathcal{P}$. In particular, it solves $(g,D,\ell)$-tester as defined above even if $t$ manipulations are made per query.

    Here, we focus on the first $m_0 = \alpha/(10\eps)$ queries made by $\cT$, and apply the offline argument.
    Define for each point $z\in D$ the event $E_z$ that $z$ is queried by $\cT$ within its first $m_0$ queries.
    Next, order the points $z$ by the probabilities of $E_z$ and define the set $D' \subseteq D$ consisting of the $\ell$ points with minimal probabilities (breaking ties arbitrarily).
    The event that any of these points was queried within the first $m_0$ queries is denoted by $E_{D'}$ and similarly to~\eqref{eq:online_lb_crucial_event_upper}, it satisfies   
    \[
        \Prob{}{E_{D'}} \leq \frac{\eps \cdot m_0}{\alpha} = \frac{1}{10} .
    \]
    
    We are now ready to define the (deterministic) adversarial strategy that interferes with the tester.
    The adversary uses the first $m_0$ queries to manipulate all the coordinates of $D'$, hiding all differences between $g$ and $g_{D'}$.
    An erasure adversary will simply erase these entries, whereas a corruption adversary will fix their values to those of $g$ (no matter the input). 
    Note that even the weaker fixed-rate adversary is able to execute this strategy deterministically (using an arbitrary order on $D'$), and that all entries in $D'$ can indeed be manipulated within $m_0$ queries since
     \[
        m_0 t 
        \geq \frac{\alpha}{10\eps} \cdot \frac{10\eps^2 q^n}{\alpha}
        = \eps q^n
        = \ell
        = \card{D'} .
    \]

    Let us fix the random coins of $\cT$ and analyze two executions with the online adversary above, one on input $g$ and the other on input $g_{D'}$.
    If the two executions differ in their answer, it means at least one query was answered differently in both executions, and this can only happen in the first $m_0$ queries due to the online adversary. This implies that for these random coins, the event $E_{D'}$ holds.
    Over the random choices of $\cT$, similarly to~\eqref{eq:online_lb_crucial_event_lower}, we get
    \[
        \Prob{}{E_{D'}}
        \geq \Prob{}{\cT\text{ accepts }g} - \Prob{}{\cT\text{ accepts }g_{D'}}
        \geq \frac{1}{3} ,
    \]
    which leads to a contradiction.
    As this applies for infinitely many values of $n$, the impossibility result holds asymptotically.
\end{proof}

To prove~\Cref{thm:online_erasures_lb}, we simply apply the theorem for linearity with $q=2$ and $\alpha = 1/2$, as for any $n \geq 1$, the affine function $\fno = 1 + x_1$ has distance $1/2$ from linearity.

\section{Testing Over the Reals}\label{sec:real testing}

\subsection{Optimal Linearity Tester}\label{subsec:real linearity}

In this section, we prove \autoref{thm:additivity-and-linearity-real-tester}. Our proof consists of modifying one of the procedures of \cite{FlemingYoshida20}, namely \Call{query-$g$}{}; and proving that the modification both (i) reduces the number of queries, and (ii) maintains the correctness.
We start with a formal definition of the self-corrected function $g:\R^n\to\R$ from \cite{FlemingYoshida20} (described in \Cref{sec:intro-reals}). 
  Every $\bm p\in\R^n$ is radially contracted to a point $\bm p/\kappa_{\bm p}\in \ball(\bm 0,1/50)$, where the normalization factor is defined as
 \begin{equation*}\label{eq:contracting factor}
    \kappa_{\bm p} := \begin{cases}
    1, &\text{ if }\|\bm p\|_2 \leq 1/50\\
    50\|\bm p\|_2, &\text{ if } \|\bm p\|_2 > 1/50
\end{cases}
\end{equation*}
Then, for a direction $\bm x\in\R^n$ we define the ``opinion" of $\bm x$ on the value $g$ should have on $\bm p$ by 
\begin{equation}\label{eq:self-correction-direction-defn}
     g_{\bm x}(\bm p) :=  \kappa_{\bm p}\left(f\left(\frac{\bm p}{\kappa_{\bm p}}-\bm x\right)+f(\bm x)\right) .
\end{equation}
 Finally, let $g$ be the global corrector, aggregating opinions taken according to the standard Gaussian:
 \begin{equation}\label{eq:self-correction-defn}
    g(\bm p) := \maj_{\bm x\sim\cN(\bm 0,I)} g_{\bm x}(\bm p).
\end{equation}
We note that a priori a majority does not necessarily exist. However, we only use this definition in the analysis of the second phase of the tester. We rely on a result by~\cite{FlemingYoshida20} and only use $g$ in the analysis when it is well-defined.
Our modified version for \Call{query-$g$}{} is given in \autoref{alg:modified query g}. 
\begin{algorithm}[ht]
  \caption{Modified query access to the self-corrected function}\label{alg:modified query g}
  \Procedure{\emph{\Call{Query-$g$}{$\bm p$}}}{
    \Given{Query access to $f\colon \mathbb{R}^n \to \mathbb{R}$ and a point $\bm p\in \R^n$;}
    Sample $\bm x \sim \mathcal{N}(\bm 0,I)$\;
    \textbf{Return}  $\kappa_{\bm p} \left( f( \bm p/\kappa_{\bm p} - \bm x) + f(\bm x) \right)$; \Comment{This is $g_{\bm x}(\bm p)$ as in \eqref{eq:self-correction-direction-defn}}
    }
\end{algorithm}

 We focus on proving the first part of \Cref{thm:additivity-and-linearity-real-tester}, showing that there exists a tester for additivity. The tester for linearity follows from the same arguments as in \cite{FlemingYoshida20} where one uses our modified \Call{query-$g$}{} subroutine in both their testers\footnote{\cite{FlemingYoshida20} derived their linearity tester from their additivity tester, so one can use our modified version as a black box.}. 

\begin{algorithm}[ht]
  \caption{Distribution-free Additivity Tester}\label{alg:zero-mean-additivity}
  \Procedure{\emph{\Call{AdditivityTester}{$f,\mathcal{D},\eps$}}}{
    \Given{Query access to $f\colon \mathbb{R}^n \to \mathbb{R}$, sampling access to an unknown 
    distribution $\mathcal{D}$, proximity parameter $\eps> 0$;}
    \textbf{Reject} if \Call{TestAdditivity}{$f$} returns \textbf{Reject}\;
    \For{$N_{\ref{alg:zero-mean-additivity}} \gets O(1/\varepsilon)$ times}{
      Sample $\bm p \sim \mathcal{D}$ and $\bm x \sim \mathcal{N}(\bm 0,I)$\;
      \textbf{Reject} if $f(\bm p) \neq g_{\bm x}(\bm p)$ {\label{step:fp-vs-gxp}} \Comment{Where $g_{\bm x}(\bm p)=\kappa_{\bm p} \left( f( \bm p/\kappa_{\bm p} - \bm x) + f(\bm x) \right)$. as in \eqref{eq:self-correction-direction-defn}}
      }
    \textbf{Accept}.
    }
\end{algorithm}

\begin{algorithm}
  \caption{{\cite[Algorithm 2]{FlemingYoshida20}} Additivity Subroutine}\label{alg:subroutines_noisy}
  \Procedure{\emph{\Call{TestAdditivity}{$f$}}}
  {
    \Given{Query access to $f\colon \mathbb{R}^n \to \mathbb{R}$;}
    \For{$N_{\ref{alg:subroutines_noisy}} \gets O(1)$ times}{
      Sample $\bm{x},\bm{y},\bm{z} \sim \cN(\bm 0, I)$\;
      \textbf{Reject} if $f(-\bm x)\neq -f(\bm x)$\;
      
      \textbf{Reject} if $f(\bm x - \bm y) \neq f(\bm x) - f(\bm y)$\;
      
      \textbf{Reject} if $f\left(\frac{\bm x - \bm y}{\sqrt{2}} \right) \neq  f \left(\frac{\bm x -\bm z}{\sqrt{2}} \right) + f \left(\frac{\bm z - \bm y}{\sqrt 2} \right)$\; 
    }
    \textbf{Accept}.
  }
\end{algorithm}

In~\cite[Theorem 1]{FlemingYoshida20}, it is shown that $O\left(\frac{1}{\eps}\log \frac{1}{\eps}\right)$ queries suffice for testing linearity if $f:\R^n\to\R$ is additive. 
We simplify their algorithm, and slightly modify the analysis to show that $O(1/\eps)$ queries suffice instead. Our simplified tester is given in \autoref{alg:zero-mean-additivity}, where we incorporate our modified procedure \Call{query-$g$}{} into the algorithm.
First it runs \autoref{alg:subroutines_noisy} (\Call{TestAddativity}{}), which is exactly the same subroutine from \cite{FlemingYoshida20}. Then it proceeds similarly to their algorithm except that when we want to evaluate the self-corrected function $g(\bm p)$ (as defined in \eqref{eq:self-correction-defn}), we now use only one random direction ($\bm x\in\R^n$), evaluating $g_{\bm x}(\bm p)$ with only two queries to $f$.
    Our approach differs from \cite{FlemingYoshida20}'s at this juncture, as \cite{FlemingYoshida20} use the subroutine   \Call{Query-$g$}{}
  to evaluate $\{g_{\bm x_i}(\bm p),\bm x_i\sim\cN(\bm 0,I)\}$ in $O(\log (1/\eps))$ random directions ($\bm x_i$'s) as a consistency check, i.e., to ensure $g_{\bm x_i}(\bm p)$'s are all the same, and then output $g_{\bm x_1}(\bm p)$. They first showed that for random $\bm x$ the value of $g_{\bm x}(\bm p)$ is indeed the value of $g(\bm p)$ with probability at least $1/2$. Then, they used many queries to amplify the success probability to be close to $1$ (in particular, greater than $1-O(\eps)$).  Their idea was to get an accurate evaluation of $g$ at \emph{all} $O(1/\eps)$ queried points.
Our observation is that a correct evaluation of $g$ at each point with only a constant probability suffices, shaving the factor used to amplify it to $1-O(\eps)$. 
Intuitively, if a point $\bm p$ such that $f(\bm p)\neq g(\bm p)$ is queried, then it is enough to evaluate $g$ on this point correctly with a constant probability (say, $1/2$). Whenever such $\bm p$ exists, other incorrect evaluations do not matter, allowing us to skip the amplification and save queries.

To formally prove \cref{thm:additivity-and-linearity-real-tester}, our starting point is the following result by \cite{FlemingYoshida20}:
\begin{lemma}[{\cite[Lemma 8]{FlemingYoshida20}}]\label{lem:FY1}
    If \autoref{alg:subroutines_noisy}
        accepts with probability at least $1/10$, then $g$ is a well-defined, additive function on $\R^n$, and furthermore, for every $\bm p\in\R^n$,
    \[\Prob{\bm x\sim\cN(\bm 0,I)}{g(\bm p)\neq g_{\bm x}(\bm p)}<1/2.\]
\end{lemma}
We are now ready to prove the main theorem of this section.
\begin{proof}
[Proof of \autoref{thm:additivity-and-linearity-real-tester}]
We first note that the number of queries made by the tester is at most $10 N_{\ref{alg:subroutines_noisy}}+3N_{\ref{alg:zero-mean-additivity}}=O(1/\eps)$. Next, we show the correctness of the tester.
    The proof of completeness remains the same as in \cite{FlemingYoshida20}: if $f$ is linear, \autoref{alg:subroutines_noisy} always accepts, and moreover, (i) $g\equiv f$ and, (ii) for all $\bm p,\bm x\in\R^n$, $g(\bm p)=g_{\bm x}(\bm p)$, ensuring \autoref{alg:zero-mean-additivity} also accepts $f$.

For the soundness, assume $f$ is $\eps$-far from all additive functions. 
If \autoref{alg:subroutines_noisy} rejects w.p. $>2/3$, we are done. 
We next show that if this is not the case, in particular, \autoref{alg:subroutines_noisy} accepts w.p. $\ge 1/3$, then the probability that one of the iterations of \autoref{step:fp-vs-gxp} rejects is $\geq 2/3$.
We next claim that, in this case, the probability of one iteration of \autoref{step:fp-vs-gxp} to accept is at most $1-\eps/2$, therefore, repeating $N_{\ref{alg:zero-mean-additivity}}=\ceil{4/\eps}$ independent times, the probability to accept in all of them is at most $(1-\eps/2)^{N_{\ref{alg:zero-mean-additivity}}}\leq (1-\eps/2)^{4/\eps}< 1/3$.
It is only left to show that one such iteration rejects w.p. $\ge \eps/2$. Indeed, (recall that we assume \autoref{alg:subroutines_noisy} accepts w.p. $\ge 1/3$)
\begin{align*}
    \Pr[\text{\autoref{step:fp-vs-gxp} rejects}]&=\Pr_{\substack{\bm p\sim\cD \\ \bm x\sim\cN(\bm 0,I)}}[f(\bm p)\neq g_{\bm x}(\bm p)]\\
    &\geq \Pr_{\substack{\bm p\sim\cD \\ \bm x\sim\cN(\bm 0,I)}}[f(\bm p)\neq g(\bm p)\wedge g(\bm p)=g_{\bm x}(\bm p)]\\
    &= \Pr_{\bm p\sim\cD}[f(\bm p)\neq g(\bm p)] \cdot \Pr_{\substack{\bm p\sim\cD\\ \bm x\sim\cN(\bm 0,I)}}[g(\bm p)=g_{\bm x}(\bm p)\mid f(\bm p)\neq g(\bm p)]\\
    &\geq  \eps/2.
\end{align*}
Where the last inequality follows from \autoref{lem:FY1}: By \autoref{lem:FY1}, the function  $g$ is additive. Since $f$ is $\eps$-far from all additive functions, it is in particular $\eps$-far from $g$ with respect to $\mathcal{D}$. Thus, $\Pr_{\bm p\sim\cD}[f(\bm p)\neq g(\bm p)] \geq \eps$. Also, by \autoref{lem:FY1}, we have 
$\Pr_{\bm x\sim\cN(\bm 0,I)}[g(\bm p)=g_{\bm x}(\bm p)]\ge 1/2$ for all $\bm p\in \R^n$, and thus 
\[\Pr_{\substack{\bm p\sim\cD\\ \bm x\sim\cN(\bm 0,I)}}[g(\bm p)=g_{\bm x}(\bm p)\mid f(\bm p)\neq g(\bm p)]=\Expc{\bm p\sim\cD}{\Pr_{\bm x\sim\cN(\bm 0,I)}[g(\bm p)=g_{\bm x}(\bm p)]\mid f(\bm p)\neq g(\bm p)}\ge 1/2.\qedhere\]
\end{proof}

We note that applying the same modification of the algorithm and the analysis improves upon the approximate tester for additivity, given in \cite[Theorem D.1]{ABFKY23}. To state this result, we first need to define concentrated distributions: 
\begin{definition}\label{def:conc-dist}
    A distribution $\mathcal{D}$ over $\R^n$ is \emph{$(\eps, R)$-concentrated} if most of its mass is concentrated in an open, $\ell_2$-ball $\ball(\bm 0,R)$ of radius $R$ centered at the origin, that is, 
    \[\Pr_{\bm p \sim \mathcal{D}}[\bm p\in\ball(\bm 0,R)] \geq 1- \eps.
\]
\end{definition}
As noted by \cite{ABFKY23} the standard Gaussian distribution is $(0.01,2\sqrt{n})$-concentrated.

\begin{theorem}\label{thm:approx-additivity}
    Let $\alpha, \varepsilon > 0$ and $\cD$ be an unknown $(\varepsilon/4,R)$-concentrated distribution. There exists a one-sided error, $O (1/\varepsilon)$-query tester for distinguishing between the case when $f$ is pointwise $\alpha$-close to some additive function and the case when, for every additive function $h$, $\Pr_{\bm p \sim \cD}[|f(\bm p)-h(\bm p)| > O(R n^{1.5}\alpha) ] > \eps$.
\end{theorem}
Its proof is the same as that of \autoref{thm:additivity-and-linearity-real-tester}, and is thus omitted.
\subsection{Low-degree testing}\label{sec:low-degree over the reals}
We next prove \cref{thm:exact-low-degree}.
The proof uses the same intuition as the proof of \autoref{thm:additivity-and-linearity-real-tester}. 

We start by bringing in the definition of the self corrected function $g$ from \cite{ABFKY23}. 
Let $\alpha_i  :=  {(-1)}^{i+1} {\binom{d+1}{i}}$, for all $i\in\{0,1,\hdots,d+1\}$, and  $r := (3d)^{-6}$. For any $\bm p \in \ball(\bm 0, \srad)$ and $\bm q \in \mathbb{R}^n$,
\begin{equation}
    \label{eq: definition of g low degree}
    g_{\bm q}(\bm p)  :=  \sum_{i=1}^{d+1} \alpha_i \cdot f(\bm p+i \bm q).
\end{equation}

The intuition is that $g_{\bm q}(\bm p)$ is the value of the univariate, degree-$d$ polynomial at the point $\bm p$, in the direction $q$, that is uniquely defined by the $d+1$ distinct evaluations $\{f(\bm p + i \bm q) : i \in [d+1]\}$.
We define the value of $g$ in two steps, first in the small open ball $\ball(\bm 0,\srad)$, and then outside it:\footnote{We again note that majority is not guaranteed a priori, but it is guaranteed if $f$ passes the first phase of the tester w.h.p., as shown by~\cite{ABFKY23}. We only use this definition for this case.}
\begin{itemize}
    \item For points $\bm p \in \ball(\bm 0,\srad)$, 
\[ g(\bm p)  :=  \maj_{\bm q \sim \mathcal{N}(\bm 0, I)} \left[ g_{\bm q}(\bm p) \right].\]
\item For points $\bm p \not \in \ball(\bm 0, \srad)$, we define $g(\bm p)$ by interpolating the evaluations of $g$ on points within $\ball(\bm 0,\srad)$ as follows: Consider the radial line $L_{\bm 0, \bm p} = \{x \bm p : x \in \mathbb{R}\}$ and fix $d+1$ (arbitrary) ``distinguished'' points along this line $c_0,\ldots, c_d \in \mathbb{R}$ such that $c_i \bm p \in \ball(\bm 0, \srad)$ for all $i$; in \autoref{alg:subroutines} we choose $c_i = ir/((d+1) \|\bm p\|_2)$. Let $p_{\bm p}\colon \mathbb{R}^n \to \mathbb{R}$ be the degree-$d$, univariate polynomial uniquely defined by these $d+1$ points, such that $p_{\bm p}(c_i)= g(c_i \bm p)$, for every $i\in[d+1]$. The value of $g(\bm p)$ is thus defined as $p_{\bm p}(1)$.
Note that if $g$ was a degree-$d$ polynomial to begin with, then we would indeed have $p_{\bm p}(1) = g(\bm p)$.
\end{itemize}
We next present the tester (\autoref{alg:low_degree_main_algorithm}), and its subroutines (\autoref{alg:subroutines}), including the modified subroutine \Call{Query-$g$}{}  it uses:

\begin{algorithm}[ht]
\caption{Distribution-Free Low-Degree Tester}\label{alg:low_degree_main_algorithm}
\Procedure{\emph{\Call{LowDegreeTester}{$f,d,\mathcal{D},\varepsilon$}}}{
  \Given {Query access to $f\colon \mathbb{R}^n \rightarrow \mathbb{R}$, degree $d\in\mathbb{N}$, sampling access to an unknown distribution $\mathcal{D}$, and farness parameter $\eps > 0$.}
    \textbf{Reject} if \emph{\Call{CharacterizationTest}{}}  \textbf{rejects}\;
    $N_{\ref{alg:low_degree_main_algorithm}} \gets O(1/\varepsilon)$\;
    \For{$N_{\ref{alg:low_degree_main_algorithm}}$ times}{
      Sample $\bm{p} \sim \mathcal{D}$\;
      \textbf{Reject} if $f(\bm{p}) \neq$ \emph{\Call{Query-$g$}{$\bm{p}$}}; {\label{step:fp-vs-query-g-p}}
    }
    \textbf{Accept};
  }  
\end{algorithm}

\begin{algorithm}[!h]
\caption{Subroutines}\label{alg:subroutines}
[Recall $\alpha_{i} :=  (-1)^{i+1}\binom{d+1}{i}$.]\\
\Procedure{\emph{\Call{CharacterizationTest}{}}}{
$N_{\ref{alg:subroutines}} \gets O(d^2)$ \;
\For{$N_{\ref{alg:subroutines}}$ times }{
\For{$j\in \{1,\dots, d+1\}$}{
    \For{$t\in \{0,\dots, d+1\}$}{
      Sample $\bm{p}\sim\mathcal{N}(\bm{0},j^2(t^2+1) I),\bm{q}\sim\mathcal{N}(\bm{0}, I)$;\Comment{[$j^2(t^2+1)$ vs. $1$ Test.]}\\
      \textbf{Reject} if $\sum_{i=0}^{d+1}\alpha_i\cdot f(\bm{p}+i\bm{q})\neq 0$\;
      Sample $\bm{p}\sim\mathcal{N}(\bm{0},j^2 I),\bm{q}\sim\mathcal{N}(\bm{0},(t^2+1) I)$;\Comment{[$j^2$ vs. $t^2+1$ Test.]}\\
      \textbf{Reject} if $\sum_{i=0}^{d+1}\alpha_i\cdot f(\bm{p}+i\bm{q})\neq 0$\;
      }
      Sample $\bm{p},\bm{q}\sim\mathcal{N}(\bm{0},j^2 I)$;\Comment{[$j^2$ vs. $j^2$ Test.]}\\
      \textbf{Reject} if $\sum_{i=0}^{d+1}\alpha_i\cdot f(\bm{p}+i\bm{q})\neq 0$\;
    }}
    \textbf{Accept}\;
    }
\Procedure{\emph{\Call{Query-$g$}{$\bm{p}$}}}{
$\srad\gets (3d)^{-6}$\;
Sample $\bm q\sim \mathcal{N}(\bm 0, I)$\;
\If{$\bm{p}\in \ball(\bm{0},\srad)$}{
\textbf{return} $g_{\bm q}(\bm p)$;\Comment{Where $g_{\bm q}(\bm p)= \sum_{i=1}^{d+1} \alpha_i \cdot f(\bm{p}+i\bm{q})$ as in \eqref{eq: definition of g low degree}}
}

\For{$i\in \{1,\dots,d+1\}$}{
 $c_i \gets i\srad/((d+1)\Vert \bm{p}\Vert_{2})$\;
 $v(c_i)\gets g_{\bm q}(c_i\bm p)$\;
}
Let $p_{\bm p} \colon \mathbb{R}\to\mathbb{R}$ be the unique degree-$d$ polynomial such that $p_{\bm p}(c_i) = v(c_i)$  for $i \in [d+1]$\;
\textbf{return} $p_{\bm p}(1)$\;
}
\end{algorithm}

We will need to use the following results from \cite{ABFKY23}:
\begin{lemma}[{\cite[Lemma 3.1 and Corollary 3.7]{ABFKY23}}]\label{lem:main_lemma_well_definedness_and_querying_g}
    If \Call{CharacterizationTest}{} fails with probability at most $2/3$, then $g$ is a well-defined, degree-$d$ polynomial, and furthermore 
    for every $\bm p \in \ball(\bm 0,r)$,
     \[ \Pr_{\bm{q}\sim\mathcal{N}(\bm{0},I)}[g(\bm p)\neq g_{\bm{q}}(\bm{p})] \leq \frac{1}{7d}. \]
\end{lemma}
\begin{proof}[Proof of \autoref{thm:exact-low-degree}]
    The proof of completeness remains the same, as in \cite{ABFKY23}. If $f$ is a degree-$d$, $n$-variate polynomial, then $f$ restricted any line $\{\bm p+i\bm q:i\in\R\}$ is a degree-$d$, univariate polynomial. So, \Call{CharacterizationTest}{} always accepts. Moreover, $g\equiv f$ and for all $\bm p,\bm q\in \R^n$ we have $g_{\bm q}(\bm p)=f(\bm p)$, so \autoref{alg:low_degree_main_algorithm} always accepts.

    In case $f$ is $\eps$-far from all degree-$d$, $n$-variate polynomials. If \Call{CharacterizationTest}{} rejects w.p. $>2/3$, we are done. Otherwise, by \autoref{lem:main_lemma_well_definedness_and_querying_g}, we have that $g$ is degree-$d$, $n$-variate polynomial. We will prove the following:
    \begin{claim}\label{clm:query-g correct}
        For any $\bm p\in \R^n$, the subroutine \Call{Query-$g$}{$\bm p$} evaluates $g(\bm p)$ correctly with probability at least $1/2$.
    \end{claim}
    Assuming \autoref{clm:query-g correct}, the argument follows similarly to the proof of \autoref{thm:additivity-and-linearity-real-tester}: The probability of a single iteration of \autoref{step:fp-vs-query-g-p} to accept is at most $1-\eps/2$; so that, in $N_{\ref{alg:low_degree_main_algorithm}}=\lceil 4/\eps\rceil$ such (independent) iterations, the probability to accept is at most $(1-\eps/2)^{N_{\ref{alg:low_degree_main_algorithm}}} \leq (1-\eps/2)^{4/\eps}<1/3$.
    \end{proof}
    It is only left to prove the claim. 
    \begin{proof}[Proof of \autoref{clm:query-g correct}]
        We have two cases:
        \begin{itemize}
            \item If $\bm p\in \ball(0,\srad)$, then \Call{Query-$g$}{$\bm p$} return $g_{\bm q}(\bm p)$ for some random $\bm q\sim \mathcal{N}(\bm{0},I)$. By \autoref{lem:main_lemma_well_definedness_and_querying_g}, this is equals to $g(\bm p)$ with probability at least $1-\frac{1}{7d}\geq 1/2$. 
            \item If $\bm p\not\in \ball(0,\srad)$, then \Call{Query-$g$}{$\bm p$} interpolates a degree-$d$, univariate polynomial over the evaluations of $g$ in the small open ball on $d+1$ distinct points on the line $L_{\bm 0,\bm p}$. If all evaluation on the $d+1$ interpolating points are correct, then the returned value $p_{\bm p}(1)$ is indeed the correct value of $g(\bm p)$. By \autoref{lem:main_lemma_well_definedness_and_querying_g}, for each point $c_i\bm p\in\ball(\bm 0,\srad)$, we have $g_{\bm q}(c_i\bm p)\neq g(c_i\bm p)$ with probability most $1-\frac{1}{7d}$. Thus, by a union bound, we have,
            \[\Pr[\text{\Call{Query-$g$}{$\bm p$}}=g(\bm p)]\geq \Pr_{\bm q\sim \mathcal{N}(\bm 0,I)}[g(c_i\bm p)=g_{\bm q}(\bm p),\ \forall i\in[d+1]]\geq 1-\frac{d+1}{7d}\ge 1/2.\qedhere\]
        \end{itemize}
    \end{proof}
 The following theorems are all shaving off the $\log \frac{1}{\eps}$ factor from the query complexity of the corresponding theorems by \cite{ABFKY23}. Their proofs are the same as above, thus we omit them. 

The following is an improvement of the approximate tester of \cite[Theorem 1.2]{ABFKY23}. \begin{theorem}\label{thm:approx-low-degree}
    Let $d\in\mathbb{N}$, for $L>0$, suppose $f: \R^n \to \R$ is a function that is bounded in  $\ball(\bm{0}, L)$, and for $\eps \in (0,1), R>0$, let $\cD$ be an $(\eps/4, R)$-concentrated distribution.   Given $\alpha>0, \beta\geq 2^{(2n)^{O(d)}}(R/L)^d\alpha$, query access to $f$, and sampling access to $\mathcal{D}$, there is a one-sided error, $O (d^5+d^2/\varepsilon)$-query tester which distinguishes between the cases when:
\begin{itemize}
    \item $f$ is pointwise $\alpha$-close to some degree-$d$ polynomial, say $h$, i.e., $|f(\bm x)-h(\bm x)|\leq\alpha$, for every $\bm x\in\R^n$, and the case when,
    \item for every degree-$d$ polynomial $h\colon\mathbb{R}^n\to\mathbb{R}$, $\Pr_{\bm p \sim \cD}[|f(\bm p)-h(\bm p)| > \beta ] > \eps$.
\end{itemize} 
\end{theorem}
The next theorem is an improvement of the discrete low-degree tester, that works over rational lattices, \cite[Theorem 1.3]{ABFKY23}.
\begin{theorem}\label{thm:discrete-low-degree}
    For $d, B, R>0$, let $B'\geq 16\cdot\max\{n^{5/2+2d} d^{2d},B^2R^2/\sqrt{n}\}$ be a multiple of $B$. Let $\cL = \frac{1}{B}\mathbb{Z}^n$ and $\cL' = \frac{1}{B'} \mathbb{Z}^n$.
Given $\eps>0$, query access to a function $f: \R^n \to \R$, and sample access to an unknown $(\eps/4, R)$-concentrated distribution $\cD$ supported on $\cL$, there is a one-sided error,  $O(d^5+d^2/\varepsilon)$-query tester for testing whether $f$ is a degree-$d$ polynomial, or is $\eps$-far from every degree-$d$ polynomials over $\mathcal{D}$. The tester queries $f$ on points in $\cL'$.
\end{theorem}

\section*{Acknowledgment} 
This work was written while the authors were visiting the Simons Institute for the Theory of Computing. UM and EK wish to thank Sofya Raskhodnikova for fruitful discussions concerning with the online model. VA wishes to thank Philips George John for his timely help with handling formatting issues.

\short{\bibliographystyle{plainurl}}{\bibliographystyle{alpha}}
\bibliography{references}

\end{document}